\documentclass{llncs} 
\newtheorem{defn}{Definition} 
\newtheorem{df}{Definition}

\newtheorem{thm}{Theorem}
\newtheorem{lem}[defn]{Lemma}
\newtheorem{cor}{Corollary}
\newtheorem {rem}{Remark}
\newtheorem{con}{Conjecture}

\title{The algebra of row monomial matrices}

\date{2.4.2022}
\author{A.N. Trahtman\thanks{Email: avraham.trakhtman@gmail.com}
\institute{{Bar-Ilan University, Dep. of Math., 52900, Ramat Gan, Israel}}
}
\begin{document}

\maketitle

 \begin{abstract}

  A word $w$ of letters on edges of underlying graph $\Gamma$ of deterministic finite 
automaton (DFA) is called synchronizing if $w$ sends all states of the automaton to a unique state ($|R(w)|=1$.
 J. \v{C}erny discovered in 1964 a sequence of $n$-state complete DFA
possessing a minimal synchronizing word of length $(n-1)^2$.

The hypothesis, well known today as the \v{C}erny conjecture, claims that  $(n-1)^2$ 
is also precise  upper bound on the length of such a word for a complete DFA. 
The hypothesis was formulated in 1966 by Starke.
The problem has motivated great and constantly growing number of investigations and generalizations.

The class of row monomial matrices (one unit and rest zeros in every row) with some 
 non-standard operations of summation and usual multiplication is our main object. 
These matrices generate a space with respect to the mentioned operations.
Every row monomial matrix is defined by word of path on DFA. 

The proof of the conjecture is based on connection between length of words  
of row monomial matrices and dimension of the space generated by row monomial  matrices 
of prefixes of synchronizing word.  
 
\end{abstract}

\section*{Introduction}

The long and fascinating history of state machine synchronization and the problems around 
was reflected in hundreds articles.
 The problem of synchronization of finite automata is a natural one and various aspects
of this problem have been touched in the literature.
Different problems of synchronization and achievements one can find in surveys \cite{Ju}, \cite{KV}
and works \cite{B}, \cite{MS}, \cite{Tm}.

The synchronizing word limits the propagation of errors for a prefix code.
Deterministic finite automaton is a tool that helps to recognized language in a set of DNA strings.

A problem with a long story is the estimation of the minimal length of synchronizing word.
 J. \v{C}erny in 1964 \cite{Ce} found the infinite sequence of $n$-state complete DFA with shortest
synchronizing word of length $(n-1)^2$ for an alphabet of size two.
Since then, only 27 small automata of length $(n-1)^2$ for $n \leq 6$ have been added
to the single \v{C}erny sequence \cite{Tm}.

  The hypothesis, well known today as the \v{C}erny's conjecture, claims that this lower bound
on the length of the synchronizing word of aforementioned automaton is also the upper bound
for the length of synchronizing word of any $n$-state complete DFA.

\begin{con}
The deterministic complete $n$-state synchronizing automaton over alphabet $\Sigma$
has synchronizing word in $\Sigma$ of length at most $(n-1)^2$ \cite{Sta} (Starke, 1966).
  \end{con}

The problem can be reduced to automata with a strongly connected graph \cite{Ce}.

We consider a class of matrices $M_u$ of mapping induced by words $u$
in the alphabet of letters on edges of the underlying graph $\Gamma$.
The matrix $M_u$ of word $u$  belongs to the class of matrices with
one unit in every row and rest zeros  (row monomial).

Initially found upper bound for the minimal length of synchronizing word was big and
has been consistently improved over the years by different authors.
The upper bound found by Frankl in 1982 \cite{Fr} is equal to $(n^3-n)/6$.
The result was reformulated in terms of synchronization in \cite{Pin}
and repeated independently in \cite{KRS}.
The cubic estimation of the bound exists since 1982.

An attempt to prove the \v{C}erny conjecture is proposed below.

\section*{Preliminaries}
We consider a complete $n$-state DFA with
 strongly connected underlying graph $\Gamma$
 over a fixed finite alphabet $\Sigma$ of labels on edges of  $\Gamma$ of an automaton $A$.
The trivial cases $n \leq 2$, $|\Sigma|=1$ and $|A \sigma|=1$ for
$\sigma \in \Sigma$ are excluded.

The restriction on strongly connected graphs is based on \cite{Ce}.
The states of the automaton $A$ are considered also as vertices of the graph $\Gamma$.

If there exists a path in an automaton from the state $\bf p$ to
the state $\bf q$ and the edges of the path are consecutively
labelled by $\sigma_1, ..., \sigma_k$, then for
$s=\sigma_1...\sigma_k \in \Sigma^+$ let us write ${\bf q}={\bf p}s$.

Let $Px$ be the set of states ${\bf q}={\bf p}x$ for all ${\bf p}$
from the subset $P$ of states and $x \in \Sigma^+$.
Let $Ax$ denote the set $Px$ for the set $P$ of all states of the automaton.

 A word $s \in \Sigma^+ $ is called a {\it synchronizing (reset, magic, recurrent, homing)}
 word of an automaton $A$ with underlying graph $\Gamma$ if $|As|=1$.
The word $s$ below denotes minimal synchronizing word such that for a state $\bf q$ $As=\bf q$.

The states of the automaton are enumerated, the state $\bf q$  has number one.

 An automaton (and its underlying graph) possessing a synchronizing word is called {\it synchronizing}.

Let us consider a linear space generated by row monomial (one unit and rest of zeros in every row)
$n \times n$-matrices.

We connect a mapping of the set of states of the automaton made by
a word $u$ of  $n \times n$-matrix $M_u$ such that for an element $m_{i,j} \in M_u$ takes place

\centerline{$m_{i,j}$= $\cases{1, &${\bf q}_i u ={\bf q}_j$; \cr 0, &otherwise.}$}

Any mapping of the set of states of the automaton  $A$  can be presented by some  word $u$
and by a corresponding row monomial matrix $M_u$.
For instance,

 \centerline{$M_u = \left(
\begin{array}{ccccccc}
  0 & 0 & 1 & . & . & . &  0 \\
  1 & 0 & 0 & . & . & . &  0 \\
  0 & 0 & 0 & . & . & . &  1 \\
  . & . & . & . & . & . &  . \\
  0 & 1 & 0 & . & . & . &  0 \\
  1 & 0 & 0 & . & . & . &  0 \\
\end{array}\right)
$}

 Let us call the matrix $M_u$ of the mapping induced by the word $u$, for brevity, the matrix of word $u$.

$M_uM_v=M_{uv}$ \cite{Be}.

The set of nonzero columns of $M_u$ (set of second indexes of its elements) of $M_u$ is denoted as $R(u)$
of size $|R(u)|$.

For linear algebra terminology and definitions, see \cite{Ln}, \cite{Ma}.

\section{Some properties of row monomial matrices}

\begin{rem} \label{r1}
The invertible matrix $M_a$ does not change the number of units of every column of $M_u$ in its image of the product $M_aM_u$.

Every unit in the product $M_uM_a$ is the product of two units, first unit from nonzero column of $M_u$ 
and second unit from a row with one unit of $M_a$.

\end{rem}

\begin{rem} \label{r4}

The columns of the matrix $M_uM_a$ are obtained by permutation of columns $M_u$.
Some columns can be merged (units of columns are moved along
row to a common column) with $|R(ua)|<|R(u)|$.

The rows of the matrix $M_aM_u$ are obtained by permutation of rows of the matrix $M_u$.
Some of these rows may disappear and replaced by another rows of $M_u$.

\end{rem}

\begin{lem} \label{l1}

The number of nonzero columns $|R(b)|$ is equal to the rank of $M_b$.

\centerline{$|R(ua)| \leq |R(u)|$} and

\centerline{$R(au) \subseteq R(u)$.}

For invertible matrix $M_a$ we have $R(au)=R(u)$ and $|R(ua)|=|R(u)|$.

 Nonzero columns of $M_{ua}$ have units also in $M_a$.

\end{lem}

\begin{proof}
The matrix $M_b$ has submatrix with nonzero determinant having only one unit in every 
row and in every nonzero column.
Therefore $|R(b)|$ is equal to the rank of $M_b$.

The matrix $M_a$ in the product $M_uM_a$ shifts column of
$M_u$ to columns of $M_uM_a$ without
changing the column itself by Remark \ref{r4} or merging.
some columns of $M_u$.

In view of possible merged columns, $|R(ua)|\leq |R(u)|$.

Some rows of $M_u$ can be replaced in $M_aM_u$ by another row and therefore some rows 
from $M_u$ may be changed, but zero columns of $M_u$ remain in $M_aM_u$ (Remark 1).

Hence $R(au) \subseteq R(u)$ and $|R(ua)| \leq |R(u)|$.

For invertible matrix $M_a$ we have $R(au)= R(u)$  and $|R(ua)|=|R(u)|$.

Nonzero columns of $M_{ua}$ have units also in $M_a$ in view of $R(ua) \subseteq R(a)$.

\end{proof}

\begin{cor}  \label{c1}
All matrices of prefixes of synchronizing row monomial matrix $s$ also have 
at least one unit in nonzero column of $s$.
\end{cor}

\begin{cor}  \label{c3}
The invertible matrix $M_a$ keeps the number of units of any column of $M_u$
 in corresponding column of the product $M_aM_u$.
\end{cor}

\subsection{Necessary conditions of the operation of summation in the class of row monomial matrices}

\begin{lem}\label{lam}
Suppose that for row monomial matrices $M_i$  and $M$
\begin{equation}
M =\sum_{i=1}^k\lambda_i M_i. \label{lm}
\end{equation}
with coefficients $\lambda$ from $Q$.

Then the sum $\sum^k_{i=1}\lambda_i =1$ and the sum $S_j$ of values in every row $j$
of the sum in (\ref{lm}) also is equal to one.

If $\sum^k_{i=1}\lambda_iM_i=0$  then $\sum_{i=1}^k \lambda_i=0$ and $S_j=0$
for every $j$ with $M_u=0$.

If the sum $\sum^k_{i=1}\lambda_i$ in every row is not unit [zero] then
$\sum_{i=1}^k\lambda_i M_i$ is not a row monomial matrix.
\end{lem}

\begin{proof}
The nonzero matrices $M_i$ have $n$ cells with unit in the cell. Therefore, 
the sum of values in all cells of the matrix $\lambda_i M_i$ is $n \lambda_i$.

For nonzero $M$ the sum is $n$. So one has in view of
$M =\sum_{i=1}^k\lambda_i M_i$

\centerline {$n=n\sum_{i=1}^k \lambda_i$, whence $1 =\sum_{i=1}^k \lambda_i$.}

Let us consider the row $j$ of matrix $M_j$ in (\ref{lm}) and let  $1_j$ be unit in the row $j$.
The sum of values in a row of the sum (\ref{lm}) is equal to unit in the row of $M$.
So $1 =\sum_{i=1}^k \lambda_i1_i=\sum_{i=1}^k \lambda_i$.

$\sum_{i=1}^k\lambda_i M_i=0$ implies $S_j=\sum_{i=1}^k \lambda_i1_i=\sum_{i=1}^k \lambda_i=0$ for  every row $j$.

If the matrix $M=\sum_{i=1}^k\lambda_i M_i$ is a matrix of word or zero matrix then
$\sum^k_{i=1}\lambda_i \in \{0, 1\}$.
If $\sum^k_{i=1}\lambda_i\not\in \{0, 1\}$ or the sum  in ${0, 1}$ is not the same in every row 
then we have opposite case and the matrix does not belong  to the set of row monomial matrix.
\end{proof}

\begin{rem} \label{r2}
The set of row monomial matrices  with considered summation operation and zero matrix together 
with multiplication generate a space.
\end{rem}

\subsection{Linear independence and dimension of the space.}

\begin{lem}  \label {v3} 

 The set $V$ of all row monomial $n\times k$-matrices
(or $n\times n$-matrices with zeros in fixed $n-k$ columns for $k\leq n$) has at most
$n(k-1)+1$ linear independent matrices.
 \end{lem}

\begin{proof}
Let us consider distinct $n\times k$-matrices of word with at most only one nonzero cell outside the last nonzero column $k$.

Let us begin from the matrices $V_{i,j}$ with unit in $(i,j)$ cell ($j<k$) and units in ($m,k$) cells for all $m$ except $i$.
The remaining cells contain zeros.
So we have $n-1$ units in the $k$-th column and only one unit in remaining $k-1$ columns of the matrix $V_{i,j}$.
Let the matrix $K$ have units in the $k$-th column and zeros in the other columns.
There are $n(k-1)$ matrices $V_{i,j}$. Together with $K$ they belong to the set $V$.
So we have $n(k-1)+1$ matrices. For instance,

\begin{picture}(0,40)
\end{picture}
$V_{1,1}=\left(
\begin{array}{cccccccc}
  1 & 0 & 0 & . & . & 0  \\
  0 & 0 & 0 & . & . & 1  \\
  0 & 0 & 0 & . & . & 1  \\
  . & . & . & . & . & .  \\
  0 & 0 & 0 & . & . & 1  \\
  0 & 0 & 0 & . & . & 1  \\
\end{array}
\right)$
\begin{picture}(4,40)
\end{picture}
$V_{3,2}=\left(
\begin{array}{cccccccc}
  0 & 0 & 0 & . & . & 1  \\
  0 & 0 & 0 & . & . & 1  \\
  0 & 1 & 0 & . & . & 0  \\
  . & . & . & . & . & .  \\
  0 & 0 & 0 & . & . & 1  \\
  0 & 0 & 0 & . & . & 1  \\
\end{array}
\right)$
\begin{picture}(4,40)
\end{picture}
$K=\left(
\begin{array}{cccccccc}
  0 & 0 & 0 & . & . & 1 \\
  0 & 0 & 0 & . & . & 1 \\
  0 & 0 & 0 & . & . & 1 \\
  . & . & . & . & . & . \\
  0 & 0 & 0 & . & . & 1 \\
  0 & 0 & 0 & . & . & 1 \\
\end{array}
\right)$

 The first step is to prove that the matrices $V_{i,j}$ and $K$ generate the space with the set $V$.
For arbitrary matrix $T$ of word from $V$ for every $t_{i,j} \neq 0$ and $j<k$,
let us consider the matrices $V_{i,j}$ with unit in the cell $(i,j)$ and the sum of them $\sum V_{i,j}=Z$.

The first $k-1$ columns of $T$ and $Z$ coincide.
   Hence in the first $k-1$ columns of the matrix $Z$ there is at most only one unit in any row.
 Therefore in the cell of $k$-th column of $Z$ one can find only value of $m$ or $m-1$.
The value of $m$ appears if there are only zeros
in other cells of the considered row. Therefore $\sum V_{i,j} - (m-1)K=T$.
Thus every matrix from the set $V$ is a span of $(k-1)n +1$ matrices from $V$.

It remains now to prove that the set of matrices $V_{i,j}$ and $K$ is a set of linear independent matrices.

If one excludes a certain matrix $V_{i,j}$ from the set of these matrices, then it is impossible
 to obtain a nonzero value in the cell $(i,j)$ and therefore to obtain the matrix $V_{i,j}$.
So the set of matrices $V_{i,j}$ is linear independent.
Every non-trivial linear combination of the matrices $V_{i,j}$ equal to a matrix of word has at
 least one nonzero element in the first $k-1$ columns.
Therefore, the matrix $K$ could not be obtained as a linear combination of the matrices $V_{i,j}$.
Consequently the set of matrices $V_{i,j}$ and $K$ forms a basis of the set $V$.
\end{proof}

\begin{cor}  \label {c2}
The set of all row monomial $n \times(n-1)$-matrices of words has $(n-1)^2$ linear independent matrices.

The set of row monomial $n\times n$-matrices has at most $n(n-1)+1$ linear independent matrices.

There are at most $n+1$ row monomial linear independent matrices of
words in the set of matrices with 2 nonzero columns
and at most $n$  linear independent matrices  in the set of matrices
with one common nonzero column.

 \end{cor}

\begin{cor}  \label {cs}

There exists a sequence  of length at most  $n(n-1)+1$ of distinct subspaces ordered by includion 
of row monomial  matrices for $n\times n$ automaton.
Zero matrix also presents a subspace, say, first int he sequence. Hence we have a sequence of
non-trivial  subspaces ordered by includion of length at most  $n(n-1)$.

\end{cor}

\subsection{The equation with unknown matrix $L$}

The row monomial solution $L$ of the equation
\begin{equation}
M_uL=M_s \label{ux}
\end{equation}
for matrix $M_s$ with all units in one column $\bf q$, row monomial matrix $M_u$, 
words $u, s \in \Sigma$ and $As=\bf q$. 
$M_u$ must have some units in the column  $q$.

\begin{df}
If the set of cells with units in
the column $\bf q$ of the matrix $M_v$ is a subset of the
analogous set of the matrix $M_u$ then we write

\centerline{$M_v \sqsubseteq_q M_u$}

\end{df}

\begin{lem} \label{l5}  
Every equation $M_uL=M_s$  has a solutions $L$ with at  least  $|R(u)|>0$ units in column $q$,
Every nonzero column $j$ of $M_u$ corresponds a unit in the cell $j$ of column $q$ of matrix L.

For solution $L$ with only $|R(u)|$ units in column $q$ (a minimal solution) $L \sqsubseteq_q N$
for any other solution $N$ of (\ref{ux}).

\end{lem}

\begin{proof}
The matrix $M_s$ of rank one has nonzero column of the state $\bf q$.

For every nonzero column $j$ of $M_u$ with elements $u_{i,j}=1$ and $s_{i,q}=1$ in the matrix $M_s$ 
the cell $(j,q)$ must have unit in the matrix $L$.
So the unit in the column $q$ of matrix $M_s$ is a product of
every unit from the column $j$ of $M_u$ and unit in the sell $j$ of column $q$ of $L$.

The set $R(u)$ of nonzero columns of $M_u$ corresponds the set of cells of the column $q$ 
with unit of minimal $L$ and  the minimal solution $L$ has in the column $q$ $|R(u)|$ units.

So to the column $q$ of every solution belong at least $|R(u)|$ units.
The remaining units of the solution $L_x$ belong to 
columns arbitrarily, but only one unit in a row.
The remaining cells obtain zero.

Lastly every solution $L$ is a row monomial matrix.

Zeros in the column $q$ of minimal $L$ correspond zero columns of  $M_u$.
Therefore for matrix $N$ such that $L \sqsubseteq_q N$
we have $M_uN=M_s$.
On the other hand, every solution $L$ must have units in cells of column $q$
 that correspond nonzero columns of $M_u$.

Thus minimal $L$ has $|R(u)|$ units in column $q$ and the equality $M_uL=M_uN=M_s$ is
equivalent to $L \sqsubseteq_q N$.

The matrix $M_u$ has set $R(u)$ of units in the column $q$ of minimal $L$.

\end{proof}

\section{Theorems}

\begin{thm} \label{t}

The deterministic complete $n$-state synchronizing automaton
$A$ with strongly connected underlying graph over alphabet $\Sigma$ has synchronizing word 
in $\Sigma$ of length at most $(n-1)^2$.

\end{thm}

Proof.
Let  synchronizing word $S$ have length at least $(n-1)^2$. 
We consider the set of solutions $L_i$ of $(n-1)^2-1$ first prefixes
 $R(u_i)$  of word $S$.

Assume that $|R(u)|>1$ in $(n-1)^2-1=n(n-2)$ first prefixes. In opposite case some
prefix $M_u$ is a synchronizing word.

All solutions $L_i$ (minimal and not minimal) have one unit in every row 
as monomial matrices and have at least $|R(u)|>1$ units in column $q$. 
The remaining units can be placed free by Lemma \ref{l5}. We restrict
this placing only to first $n-1$ columns.

By Corollary \ref{c2} of Lemma \ref{v3},  $n\times (n-1)$-matrix can allocate at most 
$n\times(n-2)+1=(n-1^2$ linear independent matrices. We consider below a possible 
way of allocation of  $n(n-2)$ matrices  $L_i$  of first prefixes of $S$ 
with $|R(u)|\geq 2$ and also matrix $S$.

For all  $n(n-2)$ cells in columns from two to $n-1$ we choose one unit
from free placing of  $L_i$ and fill by them all  $n(n-2)$ cells, one unit in one cell.
Each such unit of the matrix, if necessary, can be shifted along its column. 
The matrix $L_i$ remains to be solution of equation (\ref{ux}) for $M_i$.

It is desirable to start the process with matrices $L_i$ with large  $|R(u)|$.

Then the remaining units of every solutions $L_i$ are allocated in column $q$,
also every unit in its own row.
 So we obtain from $L_i$ corresponding not minimal solution.
Hence all new solutions $L_i$ are linear independent because all matrices 
have cells with its own  unit.

For the same reason, the matrix $M_s$ cannot be a linear combination of the matrices 
$L_i$, and therefore it can be added to the list of linearly independent $L_i$.

The solutions $L_i$ correspond sequence of $n(n-2)$ prefixes $u_i$ of matrices with $|R(u_i)|>1$. 
Corresponding  $n(n-2)$ $L_i$ are linear independent and number of them except $M_s$ is maximal.
Hence for  matrices of next words $u_i$ $|R(u_i)|=1$.

By Corollary \ref {cs}, there exists a sequence  of length at most  $n(n-1)$ 
of distinct non-trivial subspaces ordered by includion of row monomial  matrices
 for $n\times n$ automaton.

Consequently corresponding matrix $M_u$ of word $u_{(n-1)^2}$ of length $(n-1)^2$ 
(or earlier) is synchronizing.

\begin{thm}\label{t2}
The deterministic complete $n$-state synchronizing automaton
$A$ with underlying graph
over alphabet $\Sigma$ has synchronizing word in $\Sigma$ of length at most $(n-1)^2$.
\end{thm}
Follows from Theorem \ref{t} because the restriction for strongly connected graphs
can be omitted due to \cite{Ce}.


\begin{thebibliography}{99}
\bibitem{AW} R.L. Adler, B. {\it Similarity of automorphisms of the torus}, Memoirs of the Amer. Math. Soc., 98, Providence, RI,  1970.
\bibitem{Be} M.-P. B{\'e}al, {\it A note on \v{C}erny Conjecture and rational series}, technical report, Inst. Gaspard Monge, Univ. de Marne-la-Vallee, 2003.
\bibitem{Ce} J. \v{C}erny, {\it Poznamka k homogenym eksperimentom s konechnymi automatami}, Math.-Fyz. \v{C}as., 14(1964), 208-215.
 \bibitem{Fr} P. Frankl, {\it An extremal problem for two families of sets.} Eur. J. Comb., 3(1982), 125-127.
\bibitem{Ju} H. Jurgensen, Synchronization. Inf. and Comp. 206(2008), 9-10, 1033-1044.
\bibitem{KRS}  A.A. Kljachko, I.K. Rystsov, M.A. Spivak, An extremely combinatorial problem
connected with the bound on the length of a recurrent word in an automata. Kybernetika. 2(1987), 16-25.
\bibitem{KV} J. Kari, M. V. Volkov, {\it \v{C}erny's conjecture and the road coloring problem}. Handbook of Automata, 2013.
\bibitem{Ln} P. Lankaster, {\it Theory of matrices}, Acad. Press, 1969.
\bibitem{Ma} A. I. Malcev, {\it Foundations of linear algebra}.San Francisco, Freeman, 1963. (Nauka, 1970, in Russian.)
\bibitem{MS} A. Mateescu and A. Salomaa, {\it Many-Valued Truth Functions, Cerny's conjecture and road coloring}, Bulletin EATCS,  68 (1999), 134-148.
\bibitem{Pin} J.E. Pin, {\it On two combinatorial problems arising from automata theory}. Annals of Discrete Math., 17(1983), 535-548.
\bibitem{Sta} P. H. Starke, {\it Eine Bemerkung ueber homogene Experimente}. Elektronische Informationverarbeitung und Kybernetik, 2(1966), 257-259.
\bibitem{B} A.N. Trahtman, Synchronizing Road Coloring,{\it Fifth Ifip Int. Conf., TCS-WCC}, 2008, 43-53.
\bibitem{Tm} A.N. Trahtman. .Some new Features of Row Monomial Matrices for the Study of DFA. Recent Advances in Math. Res. and CS. V 9, 2022, 126-137.





\end{thebibliography}
 \end{document}